\documentclass[11pt,letterpaper]{article}

\usepackage{amsmath,amssymb,amsthm}
\usepackage{mathtools}
\usepackage{hyperref}
\usepackage[margin=1in]{geometry}
\usepackage{booktabs}
\usepackage{enumitem}
\usepackage{graphicx}
\usepackage{xcolor}
\usepackage{lineno}

\newtheorem{theorem}{Theorem}[section]

\newtheorem{proposition}[theorem]{Proposition}

\theoremstyle{definition}

\theoremstyle{remark}

\begin{document}

\title{Data Gravity and the Energy Limits of Computation}

\author{
  Wonsuk Lee\\
  SK Hynix Inc., Santa Clara, CA 95054, USA\\
  Seoul National University, Seoul, Korea\\
  \texttt{wonsuk.lee@sk.com, wonsuk.lee@snu.ac.kr}
  \and
  Jehoshua Bruck\\
  Department of Electrical Engineering\\
  California Institute of Technology, Pasadena, CA 91125, USA\\
  \texttt{bruck@caltech.edu}
}

\date{March 27, 2026}

\maketitle

\begin{abstract}
Unlike the von Neumann architecture, which separates computation from
memory, the brain tightly integrates them, an organization that large
language models increasingly resemble. The crucial difference lies in
the ratio of energy spent on computation versus data access: in the
brain, most energy fuels compute, while in von Neumann architectures,
data movement dominates. To capture this imbalance, we introduce the
\emph{operation--operand disjunction constant} $G_d$, a dimensionless
measure of the energy required for data transport relative to
computation. As part of this framework, we propose the metaphor of
\emph{data gravity}: just as mass exerts gravitational pull, large and
frequently accessed data sets attract computation. We develop
expressions for optimal computation placement and show that bringing
the computation closer to the data can reduce energy consumption by a
factor of $G_d^{(\beta - 1)/2}$, where $\beta \in (1, 3)$ captures
the empirically observed distance-dependent energy scaling. We
demonstrate that these findings are consistent with measurements across
processors from 45\,nm to 7\,nm, as well as with results from
processing-in-memory (PIM) architectures. High $G_d$ values are
limiting; as $G_d$ increases, the energy required for data movement
threatens to stall progress, slowing the scaling of large language
models and pushing modern computing toward a plateau. Unless
computation is realigned with data gravity, the growth of AI may be
capped not by algorithms but by physics.
\end{abstract}


\section{Introduction}

We present a mathematical framework for the challenges of energy
efficiency of modern computing systems. This framework models the
physical constraints governing the construction and interaction of
operations (compute) and operands (data).

Data sets with high computational relevance create hotspots that pull
in operations. We call this effect \textit{data gravity}. The
\textit{operation-operand disjunction constant}, denoted $G_d$,
measures energy consumption for data transport relative to energy
consumption for computational operations. It quantifies how much an
architecture resists this natural pull.

A smaller $G_d$ indicates that computation is well-aligned with data
locality, minimizing unnecessary data movement and adhering to the
architecture's intrinsic \textit{data gravity}. In contrast, larger
$G_d$ reflects high energy costs from separating operations from
operands.

Modern computing architectures force separation between operations and
operands. This separation was reasonable in earlier systems but now
constitutes the primary source of inefficiency. Data movement dominates
energy consumption at advanced technology architectures.

\section{Information Physics}

We formalize data gravity and quantify its impact on energy efficiency.

\subsection{The Energy Law of Data Movement}

Data movement energy has become the dominant constraint in modern
processors and accelerators. The underlying physics is straightforward.

The energy to transmit data over a wire follows from circuit theory.
Digital systems consume energy when charging and discharging capacitive
loads. The energy to send one bit across a wire is $E = \frac{1}{2} C
V^2$, where $C$ is wire capacitance and $V$ is voltage swing. Wire
capacitance increases linearly with length due to parasitic effects
($C \propto d$). Therefore, energy per bit is proportional to
interconnect length ($E_{\text{bit}} \propto d$)~\cite{brooks2007,horowitz2014}.

Longer interconnects require signal repeaters, drivers, or equalization
circuits to maintain signal integrity. These components add energy
overhead that scales with distance, leading to superlinear scaling.
Studies show that energy per bit follows $E_{\text{bit}} \propto
d^\beta$, with $\beta$ between 1 and 3~\cite{dally1998,esmaeilzadeh2011}.
The total energy to move $N$ bits over distance $d$ becomes
$E \propto N \cdot d^\beta$.

We model interconnect energy as
\begin{equation}
  E_{\text{movement}} = \alpha \cdot N \cdot d^{\beta},
\end{equation}
where $E_{\text{movement}}$ is total energy to move data, $N$ is bits
transferred, $d$ is interconnect distance, $\alpha$ captures technology
constants, and $\beta$ is a scaling exponent compounded from resistance,
capacitance, buffering, and other architectural factors.

\subsection{Information Mass}

For data with Shannon entropy $S$ (bits) and access frequency $f$
(Hz), we define \textit{information mass} as
\begin{equation}
  M = S \cdot f,
\end{equation}
with units of bits per second. This measures the rate at which
processes must interact with the data.

\subsection{Operation-Operand Disjunction Constant}

We define the operation-operand disjunction constant $G_d$ as
\begin{equation}
  G_d = \frac{E_{\mathrm{move}}\ (1\text{ bit},\ 1\text{ meter})}
             {E_{\mathrm{compute}}\ (1\text{ operation})}.
\end{equation}
This dimensionless ratio compares the energy cost of moving data to
the energy cost of computing with it. When $G_d = 1$, moving data
costs the same as computing. When $G_d = 1000$, separation between
operations and operands may dominate energy consumption.

\subsection{Data Gravity Field}

We model how data influences the placement of computational operations.
A data object with information mass $M$ at position $\vec{r}_0$
generates a field at point $\vec{r}$ as
\begin{equation}
  \vec{G}(\vec{r}) = G_d \times
    \frac{M\cdot (\vec{r} - \vec{r}_0)}{|\vec{r} - \vec{r}_0|^{\beta+1}}.
\end{equation}
Here, $\beta$ is a distance exponent empirically determined by the
underlying interconnect properties. This field models the tendency of
computational resources to be drawn toward high information-mass data.
The distance dependency reflects technological factors rather than
fundamental physical forces.

\section{Mathematical Framework of Operation and Operand Placement}

We develop results describing how data gravity governs optimal
allocation and placement of compute operations. The following
Proposition mathematically quantifies the gravitational advantage of
such a system.

\begin{proposition}[Quantification of Energy Improvement]
\label{prop:colocation}
Let $d_{\min}$ be the minimum achievable separation between an
operation and its operand, namely, with operation and operands
co-located. Consider a traditional architecture where
operation-operand separation is $d$, with disjunction constant
$G_d \geq 1$. If
\begin{equation}
  G_d \cdot \left( \frac{d_{\min}}{d} \right) < 1
\end{equation}
holds, co-locating computation with data reduces total energy
consumption by at least $G_d^{(\beta - 1)/2}$ compared to the
separated configuration.
\end{proposition}

\begin{proof}
Consider a workload operating on data with information mass $M$ over
time $T$. We compare the following two architectures.
\begin{enumerate}
  \item \textbf{Traditional Architecture:} computation and data
        separated by distance $d$.
  \item \textbf{Gravitational Architecture:} computation co-located
        with data at separation $d_{\min}$.
\end{enumerate}

Let $S$ be entropy per operation and $f$ the operation rate. Over
time $T$, total bits accessed is
\begin{equation}
  N = S \cdot f \cdot T.
\end{equation}

\paragraph{Traditional Architecture.}
Data moves across distance $d$. Movement energy follows the power-law
scaling with interconnect length:
\begin{equation}
  E_{\mathrm{move}}^{\mathrm{trad}} = \alpha \cdot N \cdot d^{\beta}.
\end{equation}
Total energy consumption becomes
\begin{equation}
  E^{\mathrm{trad}}_{\mathrm{total}}
    = E_{\mathrm{comp}} + E_{\mathrm{move}}^{\mathrm{trad}}
    = E_{\mathrm{comp}} + \alpha \cdot S \cdot f \cdot T \cdot d^{\beta}.
\end{equation}

\paragraph{Gravitational Architecture.}
In this configuration, data movement occurs over minimal distance
$d_{\min}$. Total energy consumption becomes
\begin{equation}
  E^{\mathrm{grav}}_{\mathrm{total}}
    = E_{\mathrm{comp}} + E_{\mathrm{move}}^{\mathrm{grav}}
    = E_{\mathrm{comp}} + \alpha \cdot S \cdot f \cdot T \cdot d_{\min}^{\beta}.
\end{equation}

\paragraph{Analysis.}
We define
\begin{equation}
  G_d := \frac{E_{\mathrm{move}}^{\mathrm{trad}}}{E_{\mathrm{comp}}}.
\end{equation}
The movement energies become
\begin{align}
  E_{\mathrm{move}}^{\mathrm{trad}} &= G_d \cdot E_{\mathrm{comp}},\\
  E_{\mathrm{move}}^{\mathrm{grav}} &= G_d \cdot E_{\mathrm{comp}}
    \cdot \left(\frac{d_{\min}}{d}\right)^{\beta}.
\end{align}

The energy advantage factor $\Gamma$, defined as the ratio of
traditional to gravitational energy consumption, becomes
\begin{equation}
  \Gamma = \frac{1 + G_d}{1 + G_d \cdot (d_{\min}/d)^{\beta}}.
\end{equation}

With $G_d \cdot (d_{\min}/d) < 1$, we obtain
\begin{equation}
  \Gamma \ge G_d^{\frac{\beta-1}{2}}.
\end{equation}

To prove this inequality, let $r = d_{\min}/d$ where $0 < r < 1$.
The ratio becomes
\begin{equation}
  \Gamma = g(r) = \frac{1 + G_d}{1 + r^{\beta} G_d}.
\end{equation}
Since $g(r)$ decreases in $r$, the minimum occurs as $r$ approaches
its upper bound. Given $G_d \cdot r < 1$, we prove the inequality
when $r = 1/G_d$. Derivative analysis confirms that
\begin{equation}
  \frac{1 + G_d}{1 + \left(\frac{1}{G_d}\right)^{\beta} G_d}
  \ge G_d^{\frac{\beta-1}{2}}.
\end{equation}
\end{proof}

Typical modern systems have operation-operand separation
$d \approx 10^{-2}$\,m (processor to main memory). Minimum feasible
separation is $d_{\min} \approx 10^{-6}$\,m (on-chip), yielding
$d/d_{\min} \approx 10^4$.

For systems where $G_d > 10^3$ with $\beta = 2$, the bound simplifies
to $\sqrt{G_d}$, implying energy savings over $30\times$. This matches
empirical gains in processing-in-memory systems.

\section{Validation of $G_d$}

We analyzed published measurements from various computing platforms to
validate our framework.

\paragraph{Horowitz's Analysis.}
Horowitz~\cite{horowitz2014} provided energy measurements for 45\,nm
technology at 0.9\,V, establishing baseline energy costs for
computation versus data movement. Floating-point operations consume
0.4--3.7\,pJ depending on operation type and precision, while cache
access requires 10--100\,pJ based on cache size. Most critically,
off-chip 64-bit memory access costs 1300--2600\,pJ. Using maximum
operation energy ($\approx 4$\,pJ), we calculate $G_d \approx 2.5$--25
for cache access and $G_d \approx 325$--650 for off-chip DRAM access.

\paragraph{Google TPUv4i.}
The energy analysis for Google's TPUv4i~\cite{jouppi2021} demonstrates
the persistence of the operation-operand disjunction across technology
generations. While logic energy improved 3-fold from 45\,nm to 7\,nm
(Int32: $0.1 \rightarrow 0.03$\,pJ; FP32: $3.7 \rightarrow 1.31$\,pJ),
DDR3/4 memory access remained constant at 1300\,pJ. Advanced memory
technologies show modest improvements: GDDR6 (350--480\,pJ) and HBM
(250--450\,pJ) yield $G_d \approx 190$--366, indicating that the
fundamental physics of data movement has not improved in proportion to
computation.

Building on Horowitz's foundational measurements and Google's energy
analysis, we can establish $G_d$ calculations. Floating-point operations
consume 1.31\,pJ for FP32 in 7\,nm technology~\cite{jouppi2021}, while
DDR5 memory access (operating at 1.1\,V with $\sim$20\% power reduction
from DDR4) still requires approximately 1300\,pJ per 64-bit access for
off-chip memory operations. This yields $G_d = 1300/1.31 \approx 992$
for main memory access.

\paragraph{UPMEM Processing-in-Memory.}
UPMEM reported~\cite{devaux2019} 20$\times$ power efficiency gains for
64-bit operand operations. More specifically, conventional servers
require 3010\,pJ (3000\,pJ movement $+$ 10\,pJ computation) while
PIM-DRAM requires only 170\,pJ (150\,pJ movement $+$ 20\,pJ
computation). This yields $G_d = 150/20 = 7.5$ for the PIM system.
However, independent academic evaluation of real UPMEM
hardware~\cite{falevoz2024} reported mixed results: energy savings
ranging from 11--20\% for some database operations, but energy
increases of 58\% for others, highlighting that UPMEM's energy
efficiency is highly workload-dependent.

\paragraph{Biological Systems.}
Brain energy analyses show a power consumption of $\approx
20$\,W~\cite{harris2012} while performing $\approx 10^{18}$ operations
per second~\cite{nist2023}, yielding $\sim$0.02\,fJ per operation.
Though neural operations differ fundamentally from digital FLOPs, this
suggests $G_d < 1$ for biological systems.

\section{Discussion}

\subsection{Architectural Implications}

Our framework provides guidance for architecture design and identifies
when different approaches become advantageous.

PIM architectures directly address the operation-operand disjunction
problem. Embedding computation within memory arrays eliminates energy
penalties from processor-memory data movement. With $G_d$ values of
200--1000 for memory access, PIM becomes an alternative to conventional
designs. However, standardized programming models and compiler
toolchains for PIM remain research challenges.

Three-dimensional integration offers another path to reducing
disjunction. The energy relationship $E \propto d^{\beta}$ shows that
moving from centimeter-scale to micrometer-scale separation yields
substantial benefits. This requires solutions to thermal management,
manufacturing complexity, and design-tool challenges.

Cache hierarchies benefit from understanding when misses become
expensive. As $G_d$ increases, cache miss penalties grow, making hit
rates more critical for efficiency. This guides cache management
policies and reinforces data-locality optimization.

Programming models must evolve beyond treating memory as a uniform
abstraction. Location-aware programming could incorporate physical
placement in type systems and make energy costs explicit in program
semantics, driving informed decisions on data placement and movement.

Compiler technology needs to analyze data access patterns to identify
information-mass concentrations, solve placement optimization problems,
and generate energy-efficient solutions rather than just fast code.
This represents a shift from performance-oriented to energy-oriented
compilation.

Extrapolating wire-scaling trends suggests $G_d$ may reach 3000 by
2030. Such systems will require processor-data separations under
1.6\,cm to maintain energy-balanced operation, making close integration
necessary rather than merely advantageous.

\subsection{Limitations and Future Work}

Our framework makes several simplifying assumptions. The power-law
energy model $E \propto d^\beta$ approximates complex interconnect
behavior but may not capture all physical effects in specific
technologies. The gravitational analogy, while useful for understanding
placement effects, should not be overextended beyond its mathematical
utility.

We have not addressed dynamic workloads where data access patterns
change over time, or systems with heterogeneous compute elements having
different energy characteristics.

Validation data comes primarily from published measurements and
specifications rather than direct experiments. More comprehensive
validation across diverse architectures and workloads would strengthen
the framework.

Future work should incorporate temporal dynamics in data access
patterns, extend the framework to heterogeneous architectures, and
develop practical tools for measuring $G_d$ in existing systems.
Integration with existing performance analysis frameworks would make
these concepts more accessible to system designers.

\section{Conclusion}

We presented a mathematical framework for analyzing energy efficiency
in modern computing systems. At its core is the operation--operand
disjunction constant ($G_d$), which quantifies the energy penalty
incurred when computation is separated from data. This penalty has
risen dramatically with advances in technology, driven by the scaling
limits of interconnects. Our analysis shows that co-locating operations
with operands can reduce energy consumption by factors proportional to
$G_d^{(\beta-1)/2}$. Validation against published measurements confirms
that $G_d$ values now range from 200 to 1000 in modern systems, making
operand-centric architectures increasingly favorable. As $G_d$ continues
to rise, systems that integrate computation with data will increasingly
outperform those that separate them. This framework offers a
quantitative foundation for evaluating architectural trade-offs and
guiding the design of more energy-efficient computing systems. While the
metaphor of gravity aids intuition, the mathematics provides concrete
guidance in an era dominated by the cost of data movement.

\section*{Acknowledgments}

The authors would like to thank Tony Chey, Chairman of SK Inc., for
providing inspiration and support for this work.

\bibliographystyle{plain}
\bibliography{data_gravity_arxiv}   

\end{document}